%% file: root_pwa_icra.tex
\documentclass[letterpaper, 10 pt, conference,onecolumn]{article} 

\usepackage{amsmath}
\usepackage{amsthm}
\usepackage{color}
\usepackage{algorithm}
\usepackage{tikz}
\usetikzlibrary{arrows}
\usepackage{graphicx}
\usepackage{authblk}
\usepackage{mathtools}
\usepackage{tikz}
\usetikzlibrary{arrows,automata}
\usepackage{amsfonts}
\usepackage[noend]{algpseudocode}
\usepackage[noadjust]{cite}
\usepackage{multirow}
\usepackage{balance}
\usepackage{url}

\title{Sampling-based Polytopic Trees for Approximate Optimal Control of Piecewise Affine Systems}
\author{Sadra Sadraddini and Russ Tedrake \thanks{The authors are with the Computer Science and Artificial Intelligence Laboratory, Massachusetts Institute of Technology, 32 Vassar st, Cambridge, MA 02139,
\texttt{\{sadra,russt\}@mit.edu}. This work was partially supported by ONR MURI 81825-10911 and the MIT Lincoln Laboratory. 
}}

\newtheorem{definition}{Definition}
\newtheorem{example}{Example}
\newtheorem{problem}{Problem}
\newtheorem{subproblem}{Subproblem}
\newtheorem{lemma}{Lemma}
\newtheorem{remark}{Remark}

\newtheorem{theorem}{Theorem}

\DeclareMathOperator{\convexhull}{{Convh}}
\DeclareMathOperator{\argmin}{{\arg \min}}

\DeclareMathOperator{\sample}{{\texttt{sample}}}
\DeclareMathOperator{\st}{{such~ that~}}

\DeclareMathOperator{\child}{{child}}
\DeclareMathOperator{\tree}{{tree}}
\DeclareMathOperator{\vol}{{Volume}}

\DeclareMathOperator{\bigM}{{M}}

\date{}

\begin{document}
\maketitle

\thispagestyle{empty}
\pagestyle{empty}

\begin{abstract}
Piecewise affine (PWA) systems are widely used to model highly nonlinear behaviors such as contact dynamics in robot locomotion and manipulation. Existing control techniques for PWA systems have computational drawbacks, both in offline design and online implementation. In this paper, we introduce a method to obtain feedback control policies and a corresponding  set of admissible initial conditions for discrete-time PWA systems such that all the closed-loop trajectories reach a goal polytope, while a cost function is optimized. The idea is conceptually similar to LQR-trees \cite{tedrake2010lqr}, which consists of 3 steps: (1) open-loop trajectory optimization, (2) feedback control for computation of ``funnels" of states around trajectories, and (3) repeating (1) and (2) in a way that the funnels are grown backward from the goal in a tree fashion and fill the state-space as much as possible. We show PWA dynamics can be exploited to combine step (1) and (2) into a single step that is tackled using mixed-integer convex programming, which makes the method suitable for dealing with hard constraints. Illustrative examples on contact-based dynamics are presented. 
\end{abstract}

\section{Introduction}
Many interesting behaviors in robotics are captured by highly nonlinear models. A prominent class is control through multiple contacts, where the robot must make and break contact with the environment in order to achieve its objectives. Examples include walking \cite{collins2005bipedal,grizzle2014models,deits2014footstep}, running \cite{grizzle2009mabel}, dexterous manipulation \cite{okamura2000overview}, and push-recovery \cite{pratt2006capture}. A popular approach to characterize contact-based dynamics is using piecewise affine (PWA) models. While accurate robot models are fully nonlinear, PWA models are reasonable approximations in the neighborhoods of nominal trajectories/points, just as linear approximations are, with the extra benefit of capturing contact dynamics \cite{valenzuela2016mixed,marcucci2017approximate}. PWA models are also popular in traffic networks \cite{mehr2017stochastic} and gene circuits \cite{de2004qualitative}.

Controlling PWA systems is difficult since the controller has to determine both the temporal order of modes and the inputs applied at each one. Completeness is important - not finding a solution when one exists is undesirable. Given an initial condition and a goal,  the (optimal) trajectory that steers the state to the goal while respecting the dynamics and state/control constraints can be obtained using mixed-integer convex programming (MICP). Time is discretized to obtain a finite number of decision variables. Given a discrete-time PWA model and a fixed time horizon (steps required to get into the goal), MICP approaches are sound and complete - they find optimal solutions if they exist. However, they come at a large computational price. Their unreliable computation time, even for obtaining a feasible solution instead of the optimal one, hinders online implementation for robotic tasks with fast dynamics. While there is a great deal of research on improving the runtime of MICP solvers, they are still orders of magnitude too slow for most robot control problems.  

An alternative is to move much of the computational burden to offline phase so the real-time implantation is a lookup table of simple control laws. However, existing approaches are not efficient even for relatively small systems. (Non-deterministic) finite-state abstractions \cite{Yordanov2012} require state/control discretization, which scales poorly in high dimensions. Synthesizing PWA control laws corresponding to stabilizing piecewise quadratic (PWQ) Lyapunov functions was studied in \cite{hassibi1998quadratic,han2017feedback}. But the state-control partition producing the control laws was assumed to be the same as those of PWA dynamics, which is very restrictive and the method may fail while other solutions exist \cite{han2017feedback}.

Multi-parametric programming \cite{dua2000algorithm} for model predictive control (MPC) results in explicit hybrid MPC schemes that also provide PWA control laws \cite{bemporad2000piecewise}. This approach is complete if the horizon is fixed. But it does not scale beyond very simple problems. A heuristic is to identify useful mode sequences in advance to fill the gap between explicit and full-blown online hybrid MPC \cite{hogan2016feedback,hogan2017reactive,marcucci2017approximate}. 
However, these approaches still suffer from the computational complexity caused by the number of integers - for those initial conditions that require long horizons, computations become prohibitive. 

Since the problem is of reachability class, it is amenable to {anytime algorithms} in sampling-based motion planning \cite{karaman2011sampling} such as rapidly exploring random trees (RRTs) and its variants \cite{lavalle2006planning}. RRTs for hybrid systems \cite{branicky2006sampling} provide little robustness understanding as the nodes in the tree correspond to points in the state-space. Therefore, there is no formal guarantee that a trajectory that deviates from the points is able to recover and achieve the goal. Moreover, PWA constraints pose a challenge for choosing an appropriate metric in exploring the state-space. The authors in \cite{tedrake2010lqr} used sums-of-square (SOS) programming to obtain regions of attractions (funnels) for (time-varying) linear-quadratic regulators (LQR) that stabilize goal-reaching trajectories, which are obtained using nonlinear optimization in advance. These regions are grown backward from the goal in a tree fashion. This technique is called \emph{LQR-trees}
and was originally introduced for smooth continuous-time systems, and also was applied to limit cycle stabilization of hybrid systems in \cite{rajasekaran2017lqr}. Transverse dynamics was studied to deal with switching surfaces \cite{manchester2011regions}. Nonlinear optimization of trajectories and funnels for LQR-trees becomes complicated for systems with many modes and state/input constraints. We desire an approach that exploits the properties of PWA systems and mitigates all the mentioned concerns.

In this paper, we propose a method that uses ideas both in hybrid MPC and funnels in LQR-trees. Our technique can be viewed as a discrete-time PWA version of LQR-trees. The main contributions of this paper are i) a framework for fusing trajectory optimization with the computation of polytopes of admissible states around them into a single MICP problem;  ii) sampling-based approach (similar to RRT/RRT* \cite{karaman2011sampling}) to grow a tree of polytopes backward from the goal such that the union of polytopes cover the state-space as much as possible. Once the tree of polytopes is computed, online implementation requires few matrix multiplications or small convex programs. We also obtain a cost-to-go function that over-approximates the optimal one. Similar to LQR-trees, we obtain probabilistic feedback coverage: as the number of samples go to infinity, the union of polytopes cover the whole region of admissible initial conditions with probability one. 

This paper is organized as follows. The problem is stated in Sec. \ref{sec_problem}. Technical details on polytopes computation and tree construction are provided in Sec. \ref{sec_compute} and Sec. \ref{sec_tree}, respectively. Implementation details are discussed in Sec. \ref{sec_control}. Examples are presented in Sec. \ref{sec_case}.

\section{Problem Formulation and Approach}
\label{sec_problem}
\subsubsection*{Notation}
\label{sec_prelim}

The set of real, non-negative real, integer numbers, and empty set are denoted by $\mathbb{R}$, $\mathbb{R}_+$, $\mathbb{N}$, and $\emptyset$, respectively. 
Given $\mathbb{S} \subset \mathbb{R}^n$ and $A \subset \mathbb{R}^{n_A \times n}$, we interpret $A\mathbb{S}$ as $\{As| s \in \mathbb{S}\}$. Set additions are interpreted in Minkowski sense. The vector of all ones is denoted by $\underbar{1}$, where the dimension is unambiguously interpretable from the context.    
A \emph{polyhedron} $\mathbb{H} \subset \mathbb{R}^n$ has the form $\mathbb{H}=\{ x \in \mathbb{R}^n | H x \le h\}$, where $H \in \mathbb{R}^{n_H \times n}, h \in \mathbb{R}^{n_H}$. All inequality relations are interpreted element-wise. A bounded polyhedron is called a \emph{polytope}.

We study discrete-time systems of the form 
\begin{equation}
\label{eq_system}
\begin{array}{c}
x_{t+1} = F(x_t,u_t), \\
\end{array}
\end{equation}
where $x_t \in \mathbb{X}, \mathbb{X} \subset \mathbb{R}^n$, is the state at time $t$, $u_t \in \mathbb{U}, \mathbb{U} \subset \mathbb{R}^m$, is the control input at time $t$, $t \in \mathbb{N}$, and $F: \mathbb{X} \times \mathbb{U} \rightarrow \mathbb{X}$ is a PWA function given as:
\begin{equation}
\label{eq_pwa}
F(x,u) =  A_i x + B_i u + c_i , (x,u) \in \mathbb{H}_i,
\end{equation}
where $n_\mathcal{M} \in \mathbb{N}$ is the number of modes, $\mathbb{H}_i, i \in \mathcal{M}, \mathcal{M}:=\{1,\cdots,n_{\mathcal{M}}\}$, construct a polytopic partition of $\mathbb{X} \times \mathbb{U}$, and $A_i \in \mathbb{R}^{n\times n}, B_i \in \mathbb{R}^{n \times m}$ and $c_i \in \mathbb{R}^n$ are constant matrices, representing affine dynamics in mode $i$. 

\begin{problem}
\label{problem_feasible}
Given a PWA system \eqref{eq_pwa} and a goal polytope $\mathbb{X}_{\text{Goal}} \subset \mathbb{X}$, find the largest set of initial conditions $\mathbb{X}_{\text{initial}} \subseteq \mathbb{X}$ and a control policy $\pi: \mathbb{X} \rightarrow \mathbb{U}$ such that all the points in $\mathbb{X}_{\text{initial}}$ are steered into $\mathbb{X}_{\text{Goal}}$ in finite time. Moreover, if multiple strategies are available, select the one that
\begin{equation}
\label{eq_cost}
\text{minimize } J:= \sum_{t=0}^{T_f} \gamma_i(x_t,u_t), 
\end{equation}
where $\gamma_i: \mathbb{H}_i \rightarrow \mathbb{R}, i \in \mathcal{M}$, $x_0$ is the initial state, $(x_t,u_t) \in \mathbb{H}_i$, $u_t=\pi(x_t), t=0,\cdots,T_f-1$, $T_f \in \mathbb{N}$, and $x_{T_f} \in \mathbb{X}_{\text{Goal}}$. 
\end{problem}
Our framework is able to accommodate a finite union of polytopes as the goal, but we stick to single polytope for brevity. The solution to Problem \ref{problem_feasible} consists of both the control policy $\pi$ and the set of admissible initial conditions $\mathbb{X}_{\text{initial}}$. Finding representations for both is difficult. However, given $x \in \mathbb{X}$ and $T \in \mathbb{N}$, the following MICP problem:
\begin{equation}
\label{eq_MPC}
\begin{array}{ll}
\min & \displaystyle \sum_{t=0}^{T} \gamma_i(x_\tau, u_\tau) \\
 \text{s.t} & x_{\tau+1}=F(x_\tau,u_\tau), (x_\tau,u_\tau) \in \mathbb{H}_i, \\
&  \tau=0,1,\cdots,T-1, x_0=x, x_T \in \mathbb{X}_{\text{goal}},
\end{array}
\end{equation}
yields the optimal control sequence $u^*_0,\cdots,u^*_{T-1}$, which is an open-loop plan. We have $x \in \mathbb{X}_{\text{initial}}$ if and only if \eqref{eq_MPC} is feasible for some $T \in \mathbb{N}$. Optimality is more subtle, as one has to check the solutions for all $T \in \mathbb{N}$ and pick the best. For some problems the optimal solutions have the smallest $T$ - an obvious example is time-optimality where $\gamma_i=1, \forall i \in \mathcal{M}$. 

Solving \eqref{eq_MPC} online is effectively a closed-loop policy, but as stated earlier, it is often too slow. We desire faster feedback laws. We develop an anytime algorithm that incrementally builds a set of initial conditions that asymptotically reaches $\mathbb{X}_{\text{initial}}$, but sacrifices strong claims on optimality. 

\section{Polytopic Trajectories}
\label{sec_compute}

In this section, we introduce the first part of our solution to Problem \ref{problem_feasible}. We propose a method for obtaining trajectories of polytopes to a set of target polytopes. We focus on solving the following subproblem throughout this section.  

\begin{subproblem}
\label{subproblem}
Given a finite number of polytopic targets $\mathbb{X}_{i,\text{target}}, i=1,\cdots,N,$ and $T \in \mathbb{N}$, find a sequence of polytopes $\mathbb{X}_0,\mathbb{X}_1,\cdots,\mathbb{X}_T,$ such that:
i) (polytope-to-polytope flow) for all $x \in \mathbb{X}_\tau, 0 \le \tau < T$, there exists $u \in \mathbb{U}$ such that $F(x,u) \in \mathbb{X}_{\tau+1}$ and $(x,u) \in \mathbb{H}_i$ for some $i \in \mathcal{M}$;
ii) (target constraint) $\exists i \in \{1,\cdots,N\}, \mathbb{X}_T \subseteq \mathbb{X}_{i,\text{target}}$.
\end{subproblem}
A special case is when the target is a single polytope, but we formulate the general case of multiple polytopes as it turns to be useful in Sec. \ref{sec_tree}. Note that Subproblem \ref{subproblem} often does not have a unique solution. We will add cost criteria later in the paper mainly in order to obtain ``large" polytopes. 

\subsection{Parameterization}
Here is the main technical idea of this paper. We characterize polytopes by affine transformations of a {pre-defined} polytope $\mathbb{P} \subset \mathbb{R}^{n_p}$, where $\mathbb{P}:=\{ x \in \mathbb{R}^{n_p} | P x \le \underbar{1} \}$. The user has to choose $\mathbb{P}$. For example, the unit cube with $n_p=n$ is a simple option; we highlight its advantages later in the paper. The polytope of possible states at time $t$ is given by:
\begin{equation}
\mathbb{X}_t = \{\bar{x}_t\} \oplus G_t \mathbb{P},
\end{equation}
where $\bar{x}_t \in \mathbb{R}^n$ and $G_t \in \mathbb{R}^{n \times n_p}$ are parameters that we search over.  Given $x \in \mathbb{X}_t$, one can compute $p(x) \in \mathbb{P}$ such that $x= \bar{x}_t + G_t p(x)$. Note that $p(x)$ may not be unique. A special case is when $n=n_p$ and $G_t$ is an invertible matrix, so $p(x)=G^{-1} (x-\bar{x}_t)$. Otherwise, given $\bar{x}_t$ and $G_t$, $p(x)$ is determined from a linear program  (with zero or some ad-hoc cost). We propose the following control law:
\begin{equation}
\label{eq_control_law}
u_t(x)=\bar{u}_t + \theta_t p(x),
\end{equation}
where $\bar{u}_t \in \mathbb{R}^m$ and $\theta_t \in \mathbb{R}^{m \times n_p}$ are parameters that, again, we search over. 
The set of all possible control inputs induced by \eqref{eq_control_law} at time $t$ is $\mathbb{U}_t:=\{ \bar{u}_t \} + \theta_t \mathbb{P}.$

\subsection{Mixed-Integer Encoding}

\subsubsection{Trajectory}

Let $i \in \mathcal{M}$ be such that $\mathbb{X}_t \times \mathbb{U}_t \subseteq \mathbb{H}_i$. In other words, we restrict that the product of each state/control polytope to lie in a single mode of the PWA system - the constraint ensuring this is discussed shortly. Using \eqref{eq_control_law}, we arrive at the following evolution for $\mathbb{X}_t$ if the mode is $i$:
\begin{subequations}
\label{eq_evolve}
\begin{equation}
\bar{x}_{t+1} = A_i \bar{x}_t + B_i \bar{u}_t + c_i,
\end{equation}
\begin{equation}
G_{t+1} = A_i G_t + B_i \theta_t.
\end{equation}
\end{subequations}
We encode \eqref{eq_evolve} using \emph{big-M method}, which is a standard procedure for translating PWA systems into mixed-integer constraints. We introduce binary variables $\delta^i_t \in \{0,1\}$, which are used in a way that $\delta_t^i$ takes $1$ if mode at time $t$ is $i$, and zero otherwise. Thus, we have the constraint:
\begin{equation}
\label{eq_sum_delta}
\Sigma_{i \in \mathcal{M}} \delta^i_t=1, t=0,\cdots,T.
\end{equation}
Eq. \eqref{eq_evolve} is encoded as follows. For all $i \in \mathcal{M}$, we have:
\begin{subequations}
\label{eq_evolve_mip}
\begin{equation}
- \bigM (1-\delta^i_t) \le \bar{x}_{t+1} - A_i \bar{x}_t - B_i \bar{u}_t - c_i \le \bigM (1-\delta^i_t),
\end{equation}
\begin{equation}
 - \bigM (1-\delta^i_t) \le G_{t+1} - A_i G_t - B_i \theta_t \le  \bigM (1-\delta^i_t),
\end{equation}
\end{subequations}
where $M$ is a sufficiently large positive number. We omit detailed discussions on big-M encoding, as they are thoroughly studied in the literature, see, e.g., \cite{Bemporad1999}.  

\subsubsection{Subset Maintenance}
In order to ensure that \eqref{eq_evolve} is sound, we need to enforce that for some $i \in \mathcal{M}$, $\mathbb{X}_t \times \mathbb{U}_t \subseteq \mathbb{H}_i$. This becomes a set of mixed-integer linear constraints due to the following lemma.

\begin{lemma}
\label{lemma_farkas}
Given a polytope $\mathbb{Y}=\{y \in \mathbb{R}^m | H_y y \le h_y\}$, $Q \in \mathbb{R}^{n \times m}$, $\bar{q} \in \mathbb{R}^m$, and polytopes $\mathbb{Z}_i=\{z \in \mathbb{R}^m | H_{z,i} z \le h_{z,i} \}, i=1,\cdots, N$, the condition $\exists i \in \{1,\cdots,N\} $ such that $ Q\mathbb{Y}+\bar{q} \subseteq \mathbb{Z}_i$ is equivalent to 
\begin{equation}
\label{eq_farkas}
\begin{array}{c}
\Lambda_i H_y = H_{z,i} Q_i, \Lambda_i h_y \le h_{z,i} \delta_i - H_y \bar{q}_i, \Lambda_i \ge 0, \\
\displaystyle
\delta_i \in \{0,1\}, \sum_{i=1}^N \delta_i =1, \sum_{i=1}^N \bar{q}_i =\bar{q}, \sum_{i=1}^N Q_i =Q.
\end{array}
\end{equation} 
\end{lemma}
The proof is in the extended version in \cite{link_github}. In case of $N=1$, Lemma \ref{lemma_farkas} is reduced to an extension of Farkas' Lemma in \cite{rakovic2007optimized}. It also can be shown that if binary variables in \eqref{eq_farkas} are relaxed to continuous variables in $[0,1]$, then \eqref{eq_farkas} is equivalent to $Q\mathbb{Y}+\bar{q} \subseteq \text{Convexhull}(\{\mathbb{Z}_i\}_{i=1,\cdots,N})$, which corresponds to the tightest binary relaxation. 
We use Lemma \ref{lemma_farkas} to encode both the target constraints and also ensuring the fact that $\exists i \in \mathcal{M} \st \mathbb{X}_t \times \mathbb{U}_t \subseteq \mathbb{H}_i$. 


The solution to Subproblem \ref{subproblem} involves numerical values for $$
\zeta:=\{\bar{x}_\tau,\bar{u}_\tau,G_\tau,\theta_\tau\}_{\tau=0,1,\cdots,T},
$$ which are obtained by solving the following optimization problem:
\begin{equation}
\label{volume}
\begin{array}{lll}
\zeta^* = & \argmin & \alpha(\zeta) \\
& \text{s.t} & \exists i \in \{1,\cdots,N\}, \mathbb{X}_T \subseteq \mathbb{X}_{i,\text{target}},\\
&  & \eqref{eq_sum_delta},\eqref{eq_evolve_mip}, \exists i \in \mathcal{M}, \mathbb{X}_i \times \mathbb{U}_i \subseteq \mathbb{H}_{i},
\end{array}
\end{equation}
where $\alpha$ is a cost function with appropriate domain. Note that if a trajectory of points $(\bar{x}_0,\bar{u}_0),\cdots, (\bar{x}_T,\bar{u}_T)$, is feasible, a trivial solution is $G_\tau=0,\theta_\tau=0, \tau=0,\cdots, T$, i.e., singletons instead of full-dimensional polytopes. We address this issue by introducing heuristics for $\alpha$ to obtain large, preferably full-dimensional, polytopes.   

\subsection{Volume Maximization}
\label{sec_volume}
We desire polytopes $\mathbb{X}_0,\cdots,\mathbb{X}_T$ that are large in the sense they cover the state-space as much as possible. The ideal is maximizing the volume of the union of polytopes, which is a nonlinear objective. Note that polytopes may overlap. When $n_p=n$, a simple heuristic is to maximize the trace of $G_0$ - a linear objective. If $\forall p \in \mathbb{P}$, $-p \in \mathbb{P}$ (symmetric set), then multiplying any column of $G_0$ by $-1$ does not change $\mathbb{X}_0$. Therefore, we can safely assume that all the diagonal terms of $G_0$ are positive without any restriction posed on the space of solutions. A drawback of trace maximization is that it may still lead to zero volume. Additionally, it often leads to sparse solutions in a way that few diagonal terms take large values, but others become zero. We found that including the following heuristics is useful in \eqref{volume} when $\mathbb{P}=[-1,1]^n$:
First, we include a weighted combination of $l_1$ and $l_\infty$ norms of the vector of diagonal terms of $G_0$ in $\alpha(\zeta)$. Second, by restricting $G_0$ to be an upper/lower triangular matrix and constraining all the diagonal terms to be greater than a small positive, non-zero volumes for $\mathbb{X}_0$ are guaranteed. Third, we include weighted summation of traces of $G_1,\cdots,G_T$ in $\alpha(\zeta)$. This promotes larger polytopes for $\mathbb{X}_1,\cdots,\mathbb{X}_T$. But there is no guarantee for this heuristic to prove useful as dynamical constraints dominate the relation between subsequent polytopes. Finally, if very small volume is reached, we reject the solution and resolve the optimization problem with a different objective. The mentioned heuristics make \eqref{volume} a mixed-integer linear program (MILP) problem. 


\begin{remark}
Maximizing the determinant of a square matrix subject to linear constraints can be cast as a semidefinite program (SDP) \cite{vandenberghe1998determinant}, which may prove useful for our application. While we leave investigation of this approach to future work, we note that it converts \eqref{volume} to mixed-integer semidefinite programming (MISDP), for which solvers are still not as mature as MILP/MIQP solvers. 
\end{remark}

\subsubsection*{Complexity}
The complexity of MICP solvers grow exponentially with number of integers involved, in general. The method introduced in this section introduces $Tn_{\mathcal{M}}$ binary variables, which is the same as \eqref{eq_MPC}. However, the price is introducing more continuous variables and constraints due to Lemma \ref{lemma_farkas}. The reward is obtaining a polytopic  family of trajectories and a set of admissible initial conditions. 



\section{Random Trees of Polytopes}
\label{sec_tree}

In this section, we provide the solution to Problem \ref{problem_feasible} by solving Subproblem \ref{subproblem} multiple times such that a tree is grown backward from the goal. The tree is formalized in Sec. \ref{sec_tree_structure}. The technique is conceptually similar to the procedure in RRT and LQR-trees, which is detailed in Sec. \ref{sec_grow}. 

\subsection{Tree Structure}
\label{sec_tree_structure}

\begin{definition}
\label{def_tree}
A directed, rooted, labeled tree (which we simply refer to as \emph{tree} in the rest of the paper) is a tuple $\mathcal{T}=(\mathcal{V},v_0,\child,\mathcal{C})$, where $\mathcal{V}$ is the set of nodes, $v_0 \in \mathcal{V}$ is the root node, $\child: \mathcal{V}\setminus \{v_0\} \rightarrow  \mathcal{V}$ is a function that maps each node to its unique successor (child) in the graph - the root does not have a child, and $\mathcal{C}: \mathcal{V} \rightarrow  \mathbb{R}$ is a cost function that maps each node to a real value. 
\end{definition}
For all nodes $v \in \mathcal{V}\setminus \{v_0\}$, a unique path toward the root exists: $$v,\child(v),\child(\child(v)),\cdots,v_0.$$  
The correspondence of the elements in Definition \ref{def_tree} with the solution to Problem \ref{problem_feasible} are as follows. Each node is a polytope in $\mathbb{X}$. We denote polytope corresponding to $v \in \mathcal{V}$ by $\mathbb{X}_v$.  The root node is the goal polytope: $\mathbb{X}_{v_0}= \mathbb{X}_{\text{goal}}$. Note that for $\mathbb{X}_{\child(v)}$, $v \in \mathcal{V}\setminus \{v_0\}$, we have already computed a control policy $\mu_v: \mathbb{X}_v \rightarrow \mathbb{U}$ such that $\left\{(x,\mu_v(x)) | x \in \mathbb{X}_v \right \} \subseteq \mathbb{H}_i$ for some $i\in \mathcal{M}$, and $\left \{A_i x +B_i \mu_v(x) + c_i| x \in \mathbb{X}_v \right\} \subseteq \mathbb{X}_{\child(v)}$. We define $\mathcal{C}(\mathbb{X}_v)$ as the worst-case cost induced by moving from a point in $\mathbb{X}_v$ to $\mathbb{X}_{\child(v)}$: $\mathcal{C}(v)= \max_{x \in \mathbb{X}_v} c_i(x,\mu(x))$. Using control law \eqref{eq_control_law}, it becomes:
\begin{equation}
\label{eq_W_convex}
\mathcal{C}(\mathbb{X}_v) = \max_{p \in \mathbb{P}} \gamma_i(\bar{x}_v+G_v p,\bar{u}_{v}+\theta_{v} p),
\end{equation}
where $\mathbb{X}_v=\{\bar{x}_v\}+G_v \mathbb{P}$, $G_{\child(v)}=A_i G_v+ B_i \theta_{v}$, and $\bar{x}_{\child(v)}=A_i\bar{x}_v+B_i\bar{u}_{v}+c_i$. We assume $\gamma_i$'s are given such that \eqref{eq_W_convex} is tractable. Note that  over-approximations of values in $\mathcal{C}$ are still valid for statements made in Sec. \ref{sec_control}. 
Once the tree is available, we recursively construct the cost-to-go (value) function $V: \mathcal{V} \rightarrow \mathbb{R}$ as:
\begin{equation}
\label{eq_value}
V(v)=\mathcal{C}(v)+V(\child(v)),
\end{equation}
where $v \in \mathcal{V} \setminus \{v_0\}$, and $V(v_0)=0$. As explained in Sec. \ref{sec_control}, the cost-to-go of nodes provides an upper-bound for the cost-to-go of states, and is a key component of the controller. The set of all states in the tree is given as $
\mathbb{X}_{\tree}= \bigcup_{v \in \mathcal{V}} \mathbb{X}_v.
$

\subsection{Growing the tree}
\label{sec_grow} 
\subsubsection*{Distance From a Polytope} First, we require a routine that computes the distance between $x \in \mathbb{X}$ and $\mathbb{X}_v=\{\bar{x}_v\}+G_v \mathbb{P}$ - it should be zero if and only if $x \in \mathbb{X}_v$. A natural candidate is the following optimization problem:
\begin{equation}
\begin{array}{lll}
d(x,\mathbb{X}_v)=& \min & \|\delta\|_{l} \\
& \text{s.t} & x+\delta=\bar{x}_v+G_v p, p \in \mathbb{P},
\end{array}
\label{eq_distance}
\end{equation} 
where $l \in \{1,2,\infty\}$ makes \eqref{eq_distance} a LP/QP. However, we desire closed-form expressions since we  implement the distance function many times both in tree construction (offline) and controller implementation (online). By assuming $\mathbb{P}=[-1,1]^n$, all full-dimensional polytopes become paralleltopes for which $x=\bar{x}_v+G_v p, p \in \mathbb{P},$ is equivalent to $\|G_v^{-1}(x-\bar{x}_v)\|_\infty \le 1$. Despite heuristics for obtaining full-dimensional polytopes, it is still possible that some polytopes have zero, or close to zero, volume. We circumvent this numerical issue by adding $\epsilon$, a small number, to the singular values of $G_v$ that are smaller than $\epsilon$, to obtain $G_v^\epsilon$. 
Define $$\mathbb{X}_v^\epsilon=\{\bar{x}_v\}+G_v \mathbb{P}.$$ It is easy to verify that $\mathbb{X}^\epsilon_v \subseteq  \mathbb{X}_v + \epsilon \mathbb{P}$. We  Introduce $$p_v(x):={G_v^\epsilon}^{-1} (x-\bar{x}_v).$$ Notice that $x \in \mathbb{X}^\epsilon_v \Leftrightarrow p_v \in \mathbb{P}$. Let $p^*_v(x):=\min\left(1,\max(p_v(x),-1)\right)$ - $\min$ and $\max$ operations are implemented row-wise - and 
\begin{equation}
\label{eq_distance_alpha}
d_v(x):=\lim_{\epsilon \rightarrow 0}\left\|G_v^\epsilon(p_v(x)-p^*_v(x))\right\|_\infty,
\end{equation}
which is basically the $l_\infty$ distance between $x$ and the point in $\mathbb{X}_v^\epsilon$ for which their $l_\infty$ distance is minimal in the space transformed by ${G_v^\epsilon}^{-1}$ (which inherits metric properties as it is a one-to-one transformation). The main advantage of the unconventional distance \eqref{eq_distance_alpha} is that it tends to cancel out the effect of $\epsilon$, and more importantly, it can be cast as elementary matrix operations (multiplications and row-wise $\min, \max$). Thus, it is efficiently implementable for a large number of polytopes in parallel by stacking \eqref{eq_distance_alpha} for all $v \in \mathcal{V}$. In case of other choices of $\mathbb{P}$, one may need to stick to \eqref{eq_distance}, which may be too slow for some applications.    

 
\subsubsection*{Sample and Rejection} We assume we are given a routine \texttt{sample} which randomly selects points from $\mathbb{X}$ with total support, and is repeated in an independent and identically distributed (i.i.d.) fashion. For sampling from polytopes, we used hit and run polytopic sampler in \cite{mete2012pattern}. We select points from $\mathbb{X} \setminus \mathbb{X}_{\tree}$ by rejecting samples in $\mathbb{X}_{\tree}$. Note that $x \in \mathbb{X} \setminus \mathbb{X}_{\tree} $ iff $\{d_v(x)=0| v \in \mathcal{V}\} = \emptyset$. 

\subsubsection{Tree initialization}
We initially solve Subproblem \ref{subproblem} with no constraints for $\bar{x}_0$ and focus on obtaining large volumes for polytopes. There is a trade-off in choosing $T$. If $T$ is small, then computations are faster, but we obtain fewer polytopes. Larger $T$ leads to larger problem size, and often (but not necessarily) finds larger polytopes. Once a solution to Subproblem \ref{subproblem} is obtained, we initialize the tree by adding nodes $v_0,\cdots,v_{T-1}$ corresponding to $\mathbb{X}_0,\cdots,\mathbb{X}_{T-1}$, and set $\child(v_{\tau+1})=v_{\tau}, \tau=0,\cdots,T-1$. The cost and value functions are constructed unambiguously. 

\subsubsection{Tree Growth}
The procedure is outlined in Algorithm \ref{alg_vanilla}. 
A heuristic that is useful and is essential in ensuring probabilistic coverage in Theorem. \ref{theorem_coverage} is $\bar{x}_0=x_{\sample}+\eta$ in line 5:, where $\eta_{k} \in \beta {\eta_{\max}}_{k} \mathbb{P}$ (subscript $k$ stands for $k$'th cartesian direction), and ${\eta_{\max}}_{k}$ a positive number such that $$\mathbb{X} \subseteq \prod_{k=1}^n {\eta_{\max}}_k [-1,1].$$ This heuristic allows moving from the $x_{\sample}$ to gain feasibility/larger polytopes. We randomly select $\beta \in [0,1]$. Note that we lose the guarantee that $\mathbb{X}_0 \not \subseteq \mathbb{X}_{\tree}$. It suffices to reject solutions if the centroid and (all or some of) vertices of $\mathbb{X}^*_0$ are in $\mathbb{X}_{\tree}$, which are rapidly checked using \eqref{eq_distance_alpha}. As $\mathbb{X}_{\tree}$ gets larger, we may bias $\beta$ toward smaller values to cut the corners. Moreover, as the MILP for line 2: gets larger, it is beneficial to split the target set into  smaller polytypic clusters so we solve multiple smaller MILPs. Notice that our tree extension routine is also governed by MICP. Finally, we add nodes corresponding to the obtained polytopes to the tree $\mathcal{T}$. The child, cost and value functions are updated accordingly. The default criteria in Line 6: is $\vol(\mathbb{X}_0)>0$. 

\begin{algorithm}[t]
\caption{Random Trees of Polytopes}\label{tree}
\begin{algorithmic}[1]
\Require {Initialize $\mathcal{T}^1=(\mathcal{V},v_0,\child,\mathcal{C})$}
\For {step $2$ to \texttt{Number of iterations}} 
\State Select $x_{\sample} \in \mathbb{X} \setminus \mathbb{X}^{K-1}_{\tree},$ and $1 \le T \le T_{\max}$ 
\State Solve Subproblem \ref{subproblem} with $\bar{x}_0=x_{\sample}+\eta$ and $\mathbb{X}_{\text{target}}=\mathbb{X}_{\tree}$
\If {feasible solution exists and meets criteria} 
\State Add $v^*_{0},\cdots,v^*_{T-1}$, corresponding to $\mathbb{X}_0,\cdots,\mathbb{X}_T$, to $\mathcal{V}$, $\child(v^*_\tau)=v^*_{\tau+1}, \tau=0,\cdots,T-1$, $\child({v^*_{T-1}})=v, \mathbb{X}_T \subseteq \mathbb{X}_v$.  
\State Update cost and value functions of $\mathcal{T}^K$. 
\EndIf
\EndFor
\end{algorithmic}
\label{alg_vanilla}
\end{algorithm}
 \begin{figure}[t]
\centering
\begin{tikzpicture}
\draw[fill=red!50!] (-2, 1.5) -- (-1,2) -- (-1,2.4) -- (-2,2.4) -- (-2, 1.5) ;
\draw[fill=red!50!] (-0.5, 1.7) -- (0.1,2.6) -- (0.9,2.2) -- (0.8,1.8) -- (-0.5, 1.7) ;
\draw[fill=red!50!] (-1, 1) -- (-0.5,1.5) -- (0.5,1.5) -- (0,1) -- (-1, 1) ;
\draw[fill=red!50!] (-1.3, 0.2) -- (-0.6,0.7) -- (0.3,0.5) -- (-0.6,0.1) -- (-1.3, 0.2) ;
\draw[dashed,fill=blue!50!] (1.5, 0.2) -- (0.8,0.7) -- (0.5,0.5) -- (0.8,0.1) -- (1.5, 0.2) ;
\draw[dashed,fill=blue!50!] (1.8, 1) -- (1.1,1.5) -- (0.7,1.5) -- (0.9,1) -- (1.4, 1) ;
\draw[dashed,fill=blue!50!] (0.4, 1.8) -- (0.6,2.3) -- (0,2.2) -- (-0.3,1.8) -- (0.4, 1.8) ;
\draw[dotted,blue,fill=cyan!50!] (2.4, 0.3) -- (2.6,0.8) -- (2,1.2) -- (2,0.8) -- (2.4, 0.3) ;
\draw[dotted,blue,fill=cyan!50!] (1.3, 1.2) -- (1.1,1.45) -- (0.8,1.45) -- (1,1.1) -- (1.3, 1.2);
\draw[dotted,blue,fill=cyan!50!] (1.5, 0.6) -- (1.6,0.9) -- (1.8,0.9) -- (1.7,0.6) -- (1.5, 0.6);

\node[fill,circle,inner sep=1pt,minimum size=2pt] (a0) at (-0.6, 0.3) {};
\node[fill,circle,inner sep=1pt,minimum size=2pt] (a1) at (-0.3, 1.2) {};
\node[fill,circle,inner sep=1pt,minimum size=2pt] (a2) at (0.1, 2.1) {};
\node[fill,circle,inner sep=1pt,minimum size=2pt] (a3) at (-1.5, 2) {};

\node[fill,circle,inner sep=1pt,minimum size=2pt] (b0) at (0.7, 0.5) {};
\node[fill,circle,inner sep=1pt,minimum size=2pt] (b1) at (1, 1.3) {};
\node[fill,circle,inner sep=1pt,minimum size=2pt] (b2) at (0.1, 2) {};

\node[fill,circle,inner sep=1pt,minimum size=2pt] (c0) at (2.2, 0.8) {};
\node[fill,circle,inner sep=1pt,minimum size=2pt] (c1) at (1.6, 0.7) {};
\node[fill,circle,inner sep=1pt,minimum size=2pt] (c2) at (1.1, 1.3) {};

\draw[->] (a0) -- (a1) -- (a2) -- (a3);
\draw[dashed,->] (b0) -- (b1) -- (b2);
\draw[dotted,->] (c0) -- (c1) -- (c2);

\draw[thin,dashed] (2.5,1.6) -- (-0.7,1.6) -- (-1.8,0.3);
\draw[thin,dashed] (-0.7,1.6) -- (-0.8,2.3);

\node[] at  (-2,1) {$\mathbb{H}_1$};
\node[] at  (-0.5,2.5) {$\mathbb{H}_2$};
\node[] at  (0.8,-0.3) {$\mathbb{H}_3$};

\node[] at  (-2,0) {$\mathbb{X}$};
\node[] at  (2.3,0.5) {$x_{\sample}$};

\tikzset{vertex/.style = {shape=circle,draw,minimum size=1.5em}}
\tikzset{edge/.style = {->,> = latex'}}

\node[vertex,fill=red!50] (v0) at  (3,2) {$v_0$};
\node[vertex,fill=red!50] (v1) at  (5,2) {$v_1$};
\node[vertex,fill=red!50] (v2) at  (4,1) {$v_2$};
\node[vertex,fill=red!50] (v3) at  (4,0) {$v_3$};
\node[vertex,fill=blue!50] (v4) at (5,1) {$v_4$};
\node[vertex,fill=blue!50] (v5) at (5,0) {$v_5$};
\node[vertex,fill=cyan!50] (v6) at (6,0) {$v^*_0$};
\node[vertex,fill=cyan!50] (v7) at (6,1) {$v^*_1$};
\draw[edge] (v1) to (v0);
\draw[edge] (v2) to (v1);
\draw[edge] (v3) to (v2);
\draw[dashed,edge] (v5) to (v4);
\draw[dashed,edge] (v4) to (v1);
\draw[dotted,edge] (v6) to (v7);
\draw[dotted,edge] (v7) to (v4);
\end{tikzpicture}
\caption{A schematic illustration of polytopic tree and its extension routine}
\label{example_tikz}
\end{figure}
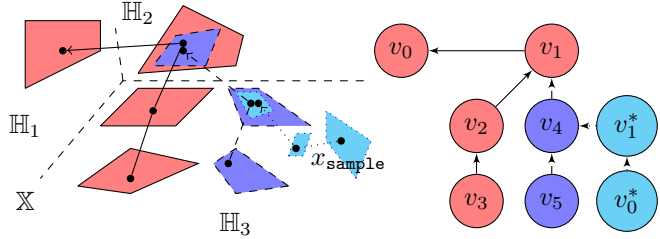

\begin{remark}
There is no optimality consideration in Algorithm \ref{alg_vanilla}, hence it may lead to unnecessarily large cost-to-go values. Inspired by the rewiring procedure in RRT* \cite{karaman2011sampling}, $\child$ function can be modified when  lower cost-to-go values for some nodes according to \eqref{eq_value} become available. Note that unlike RRT*, we do not seek asymptotic optimality guarantees since we do not resample from $\mathbb{X}_{tree}$.
\end{remark} 
\begin{theorem}
\label{theorem_coverage}
Let $\mathbb{X}_{\tree}^{K}$ be the set of states covered by the tree in the $K$'th iteration in Algorithm \ref{alg_vanilla}. Then the following property holds with probability $1$: 
\begin{equation}
\label{eq_coverage}
\lim_{K \rightarrow \infty} (\mathbb{X}_{\text{initial}} \setminus \mathbb{X}_{\tree}^{K}) = \emptyset.
\end{equation} 
\end{theorem}
\begin{proof}
see appendix
\end{proof}

\section{Control Synthesis}
\label{sec_control}

Once we obtain the tree, synthesizing the controller for Problem \ref{problem_feasible} is straightforward. In this section, we consider both the cases when the state is in $\mathbb{X}_{\tree}$ or outside of $\mathbb{X}_{\tree}$
 - for the latter we do not have formal guarantees that we can steer the state to $\mathbb{X}_{\text{goal}}$, but we provide heuristics. 

\subsubsection{$x \in \mathbb{X}_{\tree}$}
The following theorem provides the control policy. The proof follows from Sec. \ref{sec_compute} and the tree structure. 

\begin{theorem}
Let $\mu_{\mathbb{X}_{\tree}}: \mathbb{X}_{\tree} \rightarrow \mathbb{U}$ be
\begin{equation}
\label{eq_policy_tree}
\mu_{\mathbb{X}_{\tree}}(x)= \bar{u}_{v^*}+\theta_{v^*} p_{v^*}(x) 
\end{equation}
where $v^*= \argmin_{v \in V} \{ V(v) | x \in \mathbb{X}_v\}$. Then, given $x_0 \in \mathbb{X}_{\text{tree}}$,  implementing $\mu_{\tree}$ yields a trajectory $x_0,\cdots,x_T$, such that $x_\tau \in \mathbb{X}_{\text{tree}}, \tau=0,\cdots,T$,  $x_T \in \mathbb{X}_{\text{goal}}$, and the total cost $J$ in \eqref{eq_cost} is upper bounded by $V(v^*)$. 
\end{theorem}

In words, the control policy \eqref{eq_policy_tree} finds the set of polytopes which the current state belongs to, selects the one with the least cost-to-go, and implements the control law to get to its child polytope. All operations required for implementing \eqref{eq_policy_tree} are basic matrix operations - orders of magnitude faster than solving MICPs. The polytope search routine can be further accelerated using binary search trees for online implementation of PWA control laws \cite{tondel2003evaluation}.

\subsubsection{$x \in \mathbb{X} \setminus \mathbb{X}_{tree}$}
In case $ x \not \in \mathbb{X}_{tree}$, we still want to feed the system with some control input - no matter if the state can be steered toward $\mathbb{X}_{\text{goal}}$ or not. We propose the following heuristic: use \eqref{eq_distance_alpha} to find the closest polytope $\mathbb{X}_{v^*}$, and apply control input that steers $x$ toward $\mathbb{X}_{\child{v^*}}$ as much as possible by solving the following convex program:
\begin{equation}
\label{eq_policy_out}
\begin{array}{lll}
\mu_{\mathbb{X}\setminus \mathbb{X}_{\tree}}(x) = &  \arg \min & \|\delta\|_l \\
& \text{s.t} & F(x,u) + \delta \in \mathbb{X}_{\child(v^*)}, \\
\end{array}
\end{equation}
where $l \in \{1,2,\infty\}$. Note that if optimal $\delta$ is zero, then we have succeeded in getting $x$ back into the tree. We do not have any formal guarantees for the policy driven by \eqref{eq_policy_out}, but we have examples of its successful implementations (see Sec. \ref{sec_case}), even when $x$ is quite far from $\mathbb{X}_{tree}$. More elaborate formulations of \eqref{eq_policy_out} are possible in the price of higher computational cost. One can consider multiple polytopes or steps to search a wider range of getting-to-tree possibilities. Note the contrast here with full-blown online hybrid MPC as regions/policies are at least partially computed in advance.   

\input{examples}

\section{Discussion and Future work}
We believe our method is a step toward formal design of fast hybrid feedback policies for multi-contact tasks such as robotic manipulation. An issue that we largely overlooked in this paper was biased sampling, which RRT methods are shown to greatly benefit from in high dimensions  \cite{elbanhawi2014sampling,janson2015fast}. With biased sampling toward pre-computed nominal trajectories, trees can be grown locally and connected to each other in an efficient way for complex manipulation problems.    

 

\bibliographystyle{IEEEtran}

\bibliography{ana_references}

\input{appendix}
\end{document}

%% file: examples.tex
\section{Examples}
\label{sec_case}

{\bf Software.} Python scripts are publicly available in \cite{link_github}. Instructions are included to guide the user to define its own problem and use our method, or reproduce the results here. Five examples (including three shown in this paper) are currently included. The high-level details of the examples are explained here and one may refer to \cite{link_github} for full details. The cost criteria in all examples is time ($c_i=1, i \in \mathcal{M}$). 

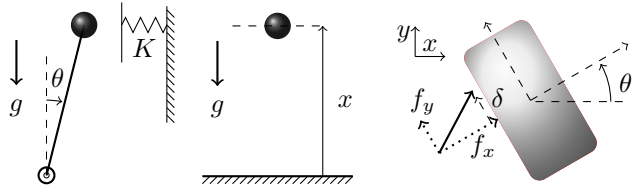
\begin{figure}[t]
\begin{tikzpicture}
\draw[thick] (0,0) -- (0.5,2);
\shade[ball color = black!99, opacity = 0.9] (0.5,2) circle (5pt);
\node[] at (0.15,1.3) {$\theta$};
\draw[thick,->] (-0.4, 1.8) -- (-0.4,1.2);
\node[] at (-0.4,0.9) {$g$};
\node[] at (1.3,1.7) {$K$};
\draw[] (1.6,0.7) -- (1.6,2.25);
\draw[] (1,1.5) -- (1,2.3);
\draw[] (1,2) -- (1.05,2.1) -- (1.15,1.9) -- (1.25,2.1) -- (1.35,1.9) -- (1.45,2.1) -- (1.55,1.9) -- (1.6,2);
\foreach \i in {1,...,15}{
\draw[] (1.6, 2.3-0.1*\i) -- (1.7,2.2-0.1*\i);}
\draw[line width=0.28mm,color=black] (0,0) circle (3pt);
\draw[line width=0.1mm] (0,0) circle (1pt);
\draw[->] (0,1) arc (90:78:1);
\draw[line width=0.1pt,dashed] (0,0) -- (0.0,1.6);
\end{tikzpicture}
~
\centering
\begin{tikzpicture}
\shade[ball color = black!99, opacity = 0.9] (0,2) circle (5pt);
\draw[thick] (-1, 0) -- (1,0);
\draw[] (-1, -0.1) -- (-0.9,0);
\foreach \i in {1,...,19}{
\draw[] (-1+\i*0.1, -0.1) -- (-0.9+\i*0.1,0);}
\draw[->] (0.6, 0) -- (0.6,2);
\draw[thick,->] (-0.8, 1.8) -- (-0.8,1.2);
\draw[dashed,] (-0.6, 2) -- (0.6,2) ;
\node[] at (0.9,1) {$x$};
\node[] at (-0.8,0.9) {$g$};
\end{tikzpicture}
~
\begin{tikzpicture}
\shade[ball color = black!10, rounded corners,fill=red,rotate around={30:(0,0)}] (0,0) rectangle (1,2);
\draw[dashed,->,rotate around={30:(0,0)}] (0.5, 1) -- (0.5,2.2);
\draw[dashed,->,rotate around={30:(0,0)}] (0.5, 1) -- (2,1);

\draw[thick,dotted,->,rotate around={30:(0,0)}] (-0.9,1) -- (0,1);
\draw[thick,dotted,->,rotate around={30:(0,0)}] (-0.9,1) -- (-0.9,1.5);
\draw[thick,->,rotate around={30:(0,0)}] (-0.9,1) -- (-0.1,1.5);
\draw[dashed,->,rotate around={30:(0,0)}] (-0.1,1) -- (-0.1,1.4);

\draw[dashed,thin] (0, 1.1) -- (1.2,1.1);

\draw[->] (1,1.1) arc (0:30:1);
\node at (1.2,1.35) {$\theta$};
\node at (-0.7,0.5) {$f_x$};
\node at (-1.5,1.1) {$f_y$};
\node at (-0.5,1.15) {$\delta$};

\draw[thin,->] (-1.6,1.7) -- (-1.6,2.1);
\draw[thin,->] (-1.6,1.7) -- (-1.2,1.7);

\node at (-1.4,1.85) {$x$};
\node at (-1.75,2) {$y$};

\end{tikzpicture}
\caption{Examples: {\bf Left:} Inverted Pendulum with Wall {\bf Middle}: Bouncing Ball {\bf Right:} Planar Pushing.}
\label{example_tikz}
\end{figure} 




\begin{example}[Inverted Pendulum with Wall]
\label{example_pendulum}
We adopt 2D example 1 from \cite{marcucci2017approximate}. Consider an inverted pendulum hitting a vertical wall at $\theta=0.1$ as in Fig. \ref{example_tikz} [Left]. The contact model is characterized by linear spring $K=1000$. Time is discretized by $0.01$. The control input $u \in [-0.4,0.4] g$ corresponds to the torque applied to the pendulum. The goal is to steer the state $(\theta,\dot{\theta})$ to the origin - a singleton. The authors in \cite{marcucci2017approximate} precomputed a set of admissible initial states in contact-free mode, denoted by $\Omega$, which is a set of states that can be driven into the origin using a linear feedback law. Explicit hybrid MPC was used to steer other states into $\Omega$. Here, we deliberately ignore exploiting the fact that $\Omega$ can be easily pre-computed, and entirely rely on our method to find $\mathbb{X}_{\text{initial}}$. Ideally, our method should recover $\Omega$. The final tree after 60 iterations is shown in Fig. \ref{fig_pendulum_color}. In comparison to the result in \cite{marcucci2017approximate}, illustrated in Fig. \ref{fig_pendulum_color} [Left], not only we recover most of $\Omega$ (and a bit beyond, as richer PWA laws are considered), but we also find a fairly large set of initial conditions in the top right of the state-space that the method in \cite{marcucci2017approximate} did not find. The reason is that the MPC horizon was limited to $10$ in \cite{marcucci2017approximate}, whereas our method does not directly suffer from short horizons. 
The green-red color spectrum in Fig. \ref{fig_pendulum_color} [Right] correspond to cost-to-go values. 
 

\begin{figure}[t]
\centering
\includegraphics[height=0.25\textwidth]{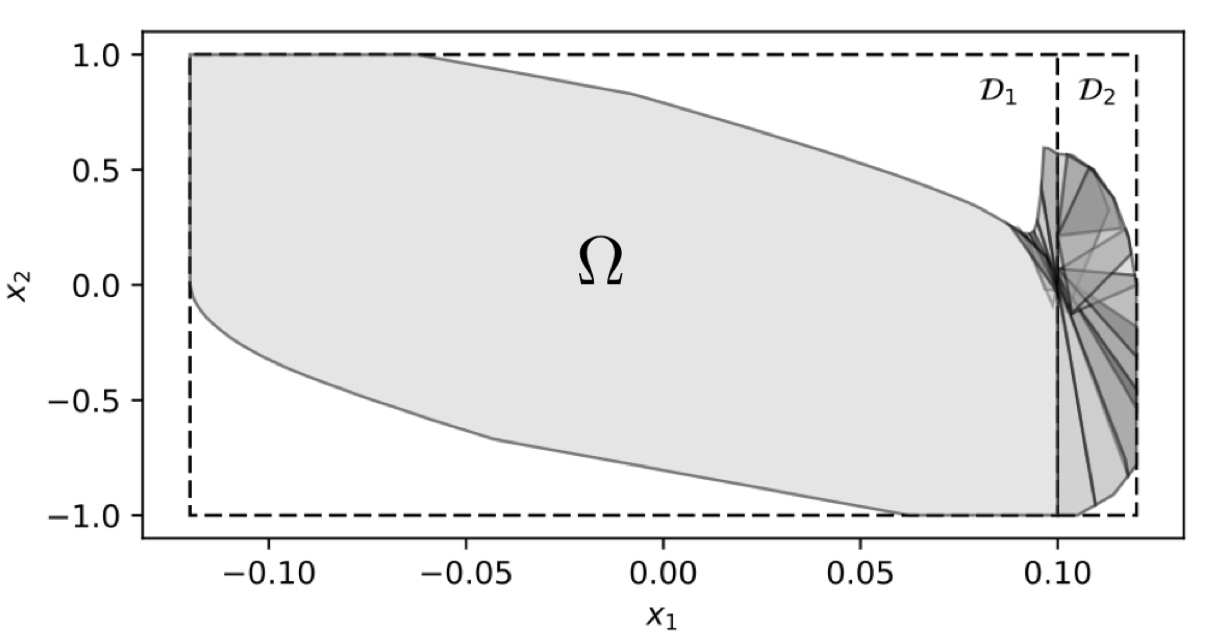}
\includegraphics[height=0.25\textwidth]{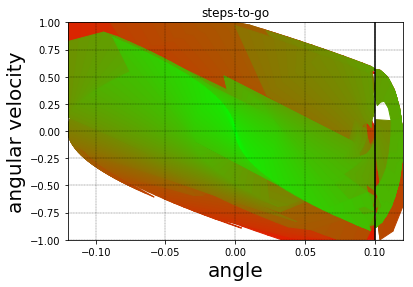}
\caption{Example 1: The results in {\bf Left:} \cite{marcucci2017approximate}, {\bf Right:} this paper.}
\label{fig_pendulum_color}
\end{figure}
\end{example}


\begin{example}[Bouncing Ball]
We consider the vertical motion of a ball falling under the gravity and ground impacts (see Fig. \ref{example_tikz} [Middle]). We have $\ddot{x}=-g+u$ with velocity sign change after hitting the ground, where $g=-9.8m/s^2$ and $u$ is the control input. We use $\Delta t=0.02s$ for time-discretization. We set $\mathbb{U}=\{u| \|u\| \le 3 m/s^2\}$, which indicates that the control actuation is not powerful enough to keep the ball from accelerating downward. The goal is to design a feedback strategy and find the set of admissible initial conditions such that the ball height and velocity reach $[1,1.2]m$ and $[-0.5,0.5]m/s$, respectively, while always satisfying the hard constraint that the velocity is within $[-5,5] m/s^2$. This problem is quite challenging because a lot of contacts may be necessary. The authors in \cite{ansari2016sequential} also consider at a similar problem, but the method is not correct-by-design as it involves heuristics. Similarities exist between this example and swinging up an inverted pendulum in \cite{tedrake2010lqr}. Both have nonlinear nature, and applying maximal control input in the direction of velocity increases energy. 

Using $T_{\max}=20$, Algorithm \ref{alg_vanilla} ``discovers" at the third iteration that by using impacts, more states can be driven into the goal. By increasing the number of iterations, nearly all states close to the origin are covered as the number of bounces are increased. Note that finding polytopes becomes more difficult as the space shrinks and polytopes close to the switching surface become small. The final tree after 100 iterations and 1049 polytopes is shown in Fig. \ref{fig_ball}. A clear disadvantage of the method in this paper versus LQR-trees in \cite{tedrake2010lqr} is caused by the discrete-time nature of the problem that leads to larger number of iterations required to fill the state-space. If the empty space between connecting polytopes is filled in a continuous-time sense, larger trees can be formed more quickly. We leave formal investigation of this issues to our future work. We randomly selected 500 points in $\mathbb{X}$ and found $359$ of them are in $\mathbb{X}^{100}_{\tree}$. By checking feasibility of \eqref{eq_MPC} for different values of $T$, we found $416$ points are drivable to the goal in less than $80$ steps - the tree has covered nearly $86$ percent of them. However, when \eqref{eq_MPC} is performed with smaller horizons, the tree has much more coverage. For instance, for $T \le 20 (30)$, only $106 (311)$ of $359$ points in $\mathbb{X}_{\tree}$ lead to feasible \eqref{eq_MPC}. 
Two sample trajectories using control  policies in \eqref{eq_policy_tree} and \eqref{eq_policy_out} (implemented on a much lesser grown tree) are shown - note the success of the latter.

\begin{figure}[t]
\centering
\includegraphics[width=0.3\textwidth]{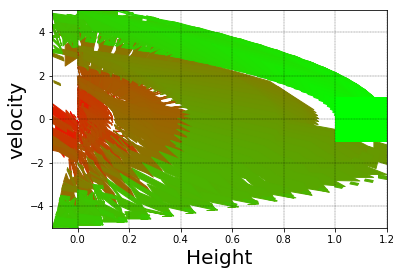}
\includegraphics[width=0.3\textwidth]{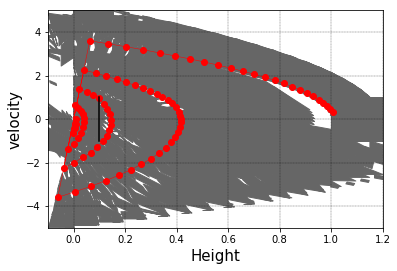}
\includegraphics[width=0.3\textwidth]{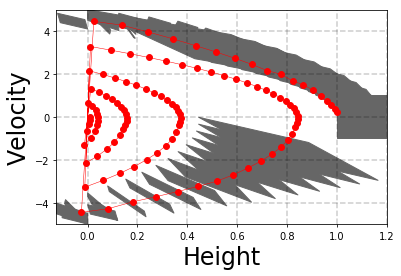}
\caption{Example 2: {\bf Left:} Cost-to-go. {\bf Middle:} $\mu_{\mathbb{X}_{\tree}}$. {\bf Right:} $\mu_{\mathbb{X} \setminus \mathbb{X}_{\tree}}$.}
\label{fig_ball}
\vspace{-0.15in}
\end{figure}

 \begin{figure}[t]
\centering
\includegraphics[width=0.32\textwidth]{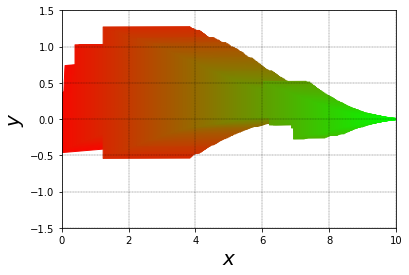}
\includegraphics[width=0.32\textwidth]{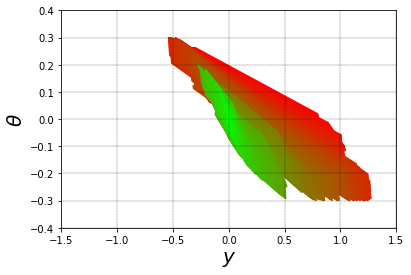}
\includegraphics[width=0.32\textwidth]{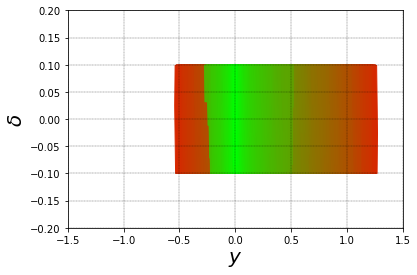}
\includegraphics[width=0.32\textwidth]{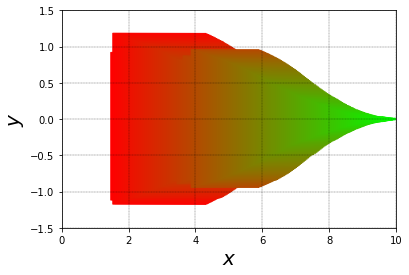}
\includegraphics[width=0.32\textwidth]{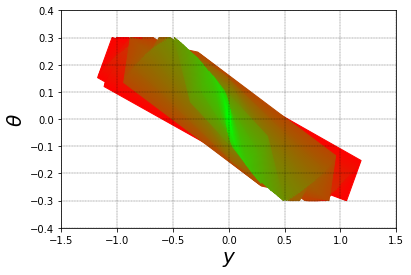}
\includegraphics[width=0.32\textwidth]{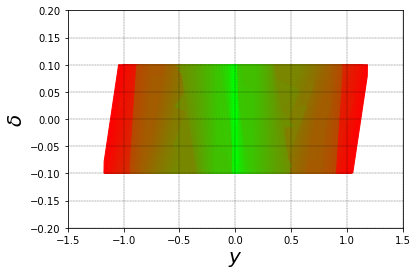}
\includegraphics[width=0.32\textwidth]{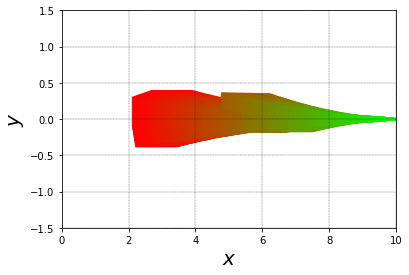}
\includegraphics[width=0.32\textwidth]{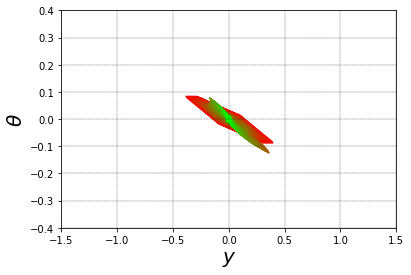}
\includegraphics[width=0.32\textwidth]{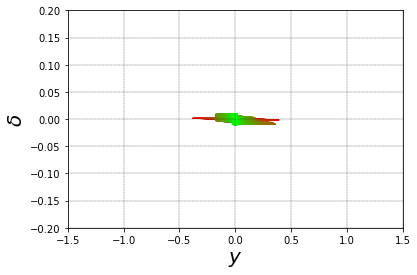}
\caption{Example 3: {\bf Top Row:} Tree projections for $T_{\max}=$ after 473 iterations (1159 polytopes). {\bf Middle Row:} Tree Projections for $T_{\max}=30$ after 16 iterations (251 polytopes). {\bf Bottom Row:} Tree projections for constrained linear model (only sticking) after 66 iterations (363 polytopes)}
\label{fig_push}
\vspace{-0.15in}
\end{figure}

 \begin{figure}[t]
\centering
\includegraphics[width=0.32\textwidth]{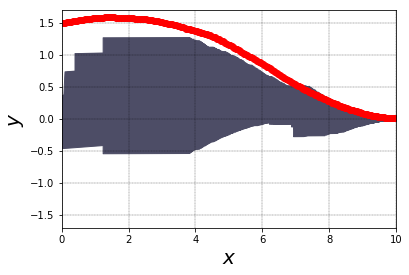}
\includegraphics[width=0.32\textwidth]{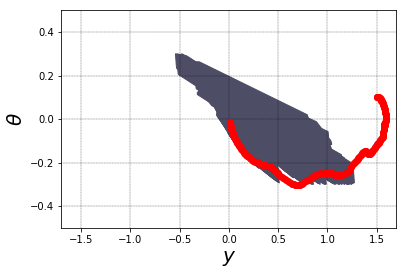}
\includegraphics[width=0.32\textwidth]{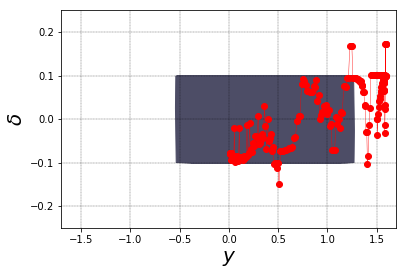}
\caption{Example 3: A sample trajectory for the controller corresponding to the first row in Fig. \ref{fig_push}}
\label{fig_push_run}
\vspace{-0.15in}
\end{figure}

\end{example}

\begin{example}[Planar Pushing] The problem and the model is adopted from \cite{hogan2017reactive}. The 6D state consists of $(x,y,\theta)$ for the box, and $(\delta,f_x,f_y)$ for the contact point and forces, as illustrated in Fig. \ref{example_tikz} [Right] - we have augmented the state with controls to construct the PWA dynamics and cells in the state-space. When the contact exists, there are three modes based on whether the contact point is fixed (pusher sticked), slides up, or slides down. The Coulomb friction coefficient between the pusher and the box is set to $\mu=0.02$, which means it is very hard to stick the pusher to the box.  The model is nonlinear, including bilinear terms for exerted wrench. We consider hybrid stabilization of a nominal trajectory and use the local PWA dynamics to compute the tree around the nominal trajectory, which is the horizontal line with $y=\theta=\delta=f_y=0, f_x=1$. We set $0.8 \le |f_x| \le 1.2, |\theta| \le 0.3, |\delta|\le 0.1,$ to approximately linearize the bilinear terms. The goal is a box of size $0.01$ around $x=10$ and other variables set to zero. 

First, we set $T_{\max}=5$ to obtain shorter branches, but more quickly. The results after 473 iterations and obtaining 1159 polytopes are illustrated in shown Fig. \ref{fig_push}. We have projected the tree on $x-y$, $y-\theta$, and $y-\delta$ planes. The asymmetry is due to the sampling nature of our solution - the tree was grown toward greater $y$'s and we expect more iterations are required to grow the tree in the other direction. A sample trajectory from $x=0,y=1.5,\theta=0.1,\delta=0,f_x=1,f_y=0$, deliberately selected from outside of the tree, is shown in Fig. \ref{fig_push_run}. It is observed that the controller manages to bring the state into the tree and remains thereafter until reaching the goal at $t=227$. As a basic numerical comparison, we found solving the hybrid MPC with horizon $T=100-200$ rates about 0.5-1.2Hz, whereas a rudimentary setup of our controller is much faster at 20-50Hz. Our solution compares to \cite{hogan2017reactive}, which uses machine learning to ``learn" the mode sequences. In contrast, our synthesis is formal and the solution is correct-by-design - at least as long as $x_0 \in \mathbb{X}_{\tree}$. 

Next, we set $T_{\max}=30$, which constructs a tree with longer branches, but the MILP for each branch is larger. The result after 16 iterations and 251 polytopes is shown in the bottom row in \ref{fig_push}. It is observed that the tree is quite symmetric this time. Finally, we expect a linear controller to perform poorly on this problem as it can not take into account mode switches. A tree for  the system constrained to remain in the sticking mode was computed and is shown in Fig. \ref{fig_push} for comparison.    
\end{example}

%% file: appendix.tex
\clearpage
\section*{Appendix}

\subsection*{Proof for Lemma 1.}
First, we prove the following result from basic convex analysis. 

\begin{lemma}
\label{prop_subset}
Given two polytopes $\mathbb{A}=\{x | H_A x \le h_A\} \subset \mathbb{R}^{n_A}, \mathbb{B}=\{x | H_B x \le h_B\} \subset \mathbb{R}^{n_B}$, and $T \in \mathbb{R}^{n_A \times n_B}$, $d \in \mathbb{R}^{n_A}$. Then $T\mathbb{B}\oplus\{d\} \subseteq \mathbb{A}$ is equivalent to:
$$\exists \Lambda \ge 0~ \st \Lambda H_B= H_A T, \Lambda h_B \le h_A - H_A d.$$
\end{lemma} 
\begin{proof}
We obtain the second relation from the first. The first condition is equivalent to $n_A$ linear program:
\begin{equation}
\label{eq_primal}
\begin{array}{ll}
\underset{x}\max & e_i' H_A (Tx + d), \\
\st & H_Bx \le h_B,
\end{array}
\le e'_i h_A, i=1,\cdots,n_A
\end{equation}
where $e_i$ is the unit vector in the $i$'th direction. By reformulating and writing the dual of each linear program in \eqref{eq_primal}, we use strong duality to have:
\begin{equation}
\label{eq_dual}
\begin{array}{ll}
\underset{z}\min & h_B' z, \\
\st & H_B' z = T'H'_A e_i,\\
& z\ge 0,
\end{array}
\le e_i h_A-e_i'H_AT, i=1,\cdots,n_A
\end{equation}
which it means $\exists z_i, i=1,\cdots,n_A$, such that
\begin{equation}
\label{eq_final}
z_i \ge 0, z'_i h_B \le h'_A e_i-T H_A e_i, z_i' H_B = e_i' H_A T, i=1,\cdots,n_A.
\end{equation}
Now define $\Lambda:=(z_1,\cdots,z_n)'$ and the proof immediately follows. 
\end{proof}

\subsubsection*{Proof for Lemma 1. (cont'd)}
Let $\delta_i=1, \delta_j=0, j=1,\cdots,N, j\neq i$. Then we have $Q_i\mathbb{Y} + q_i \subseteq \mathbb{Z}_i$ form Proposition \ref{prop_subset}. Moreover, the only feasible solution to $\Lambda_j h_q \le - H_{z,j} q_j$ is $\Lambda_j=0, q_j=0$, so $Q_j=0$ - otherwise, $Q_j \mathbb{Y}+q_i \subseteq \{0\}$ for non-zero $q_i$ and $Q_i$, which is impossible given that $\mathcal{Q}$ is not a singleton. Thus, $d=d_i, T=T_i$.

Moreover, we prove that if $\delta_i \in [0,1]$, then we have $T\mathcal{Q} \oplus\{d\} \subseteq \convexhull(\bigcup_{i=1}^N \mathcal{P}_i)$.  it follows from Proposition \ref{prop_subset} that:
$$
Q_i \mathbb{Y} \oplus\{q_i\} \subseteq \delta_i \mathbb{Z}_i.
$$ 
We take the Minkowski sum of the two sides to arrive at:
$$
\bigoplus_{i=1}^N Q_i \mathbb{Y} \oplus \{\sum_{i=1}^N q_i\} \subseteq \bigoplus_{i=1}^N \delta_i \mathbb{Z}_i,
$$
where the right hand side is equal to $\convexhull(\bigcup_{i=1}^N \mathbb{Z}_i)$. Furthermore, we have
$$
(\sum_{i=1}^N Q_i) \mathbb{Y}  \subseteq \bigoplus_{i=1}^N Q_i \mathbb{Y},
$$
which indicates $Q\mathbb{Y} \oplus\{d\} \subseteq \convexhull(\bigcup_{i=1}^N \mathbb{Z}_i)$, and the proof is complete.

\subsection*{Proof for Theorem 1.}

\begin{proof}
First, observe that $\mathbb{X}_{\tree}^{K}$ monotonically grows: $\mathbb{X}_{\tree}^{K} \subseteq \mathbb{X}_{\tree}^{K+1}$, and is upper bounded by $\mathbb{X}$, so the limit exists. Also note that all sets are closed as Problem \ref{problem_feasible} is formulated with all sets being compact, and the same holds for Subproblem \ref{subproblem} as all inequalities are non-strict in MICP formulation.   

We recall the results of explicit hybrid MPC. We know that given a fixed $T$, the explicit solution to \eqref{eq_MPC} produces a finite number of polyhedral partitions with affine feedback law in each one \cite{bemporad2000piecewise}. By varying $T$, the whole $\mathbb{X}_{\text{initial}}$ can be filled with polytopes with affine feedback law in each one. However, there is no claim about the volumes of these polyhedral partitions - some may be less than $n$-dimensional. In fact, due to hybrid dynamics and discrete-time nature, there is no guarantee that $\mathbb{X}_{\text{initial}}$ is even simply connected. It may consist of disconnected regions of zero Lebesgue measure, for which we do not have means to cover using sampling-based approaches. The key to overcome this issue lies in moving away from $x_{\sample}$ to search for feasibility, as highlighted in line 5: in Algorithm \ref{alg_vanilla}. 

First, we prove that we are able to cover the ``full-dimensional neighborhoods'' of $\mathbb{X}_{\text{initial}}$, or equivalently prove the  following property:
\begin{equation}
\label{eq_theorem_volume}
\lim_{K \rightarrow \infty} \vol(\mathbb{X}_{\text{initial}} \setminus \mathbb{X}_{\tree}^{K}) = 0.
\end{equation}
We verify \eqref{eq_theorem_volume} by showing that when $\vol(\mathbb{X}_{\text{initial}} \setminus \mathbb{X}_{\tree}^{K}) \neq 0$, then $\mathbb{X}_{\tree}^{K+1} \setminus \mathbb{X}_{\tree}^{K}$ has non-zero volume with non-zero probability. 
First, given $T_{\max}$, let $\mathbb{X}_{\text{initial}}^{T_{\max}}$ be the set of all states that can be driven into $\mathbb{X}_{\text{goal}}$ within $T_{\max}$ steps. First, we prove that \eqref{eq_coverage} holds for $\mathbb{X}_{\text{initial}}^{T_{\max}} \setminus X^K_{\tree}$. Then by induction, we prove the argument for any multiplies of $T_{\max}$, and thus $\mathbb{X}_{\text{initial}} \setminus X^K_{\tree}$. 

When $\vol(\mathbb{X}_{\text{initial}}^{T_{\max}} \setminus \mathbb{X}_{\tree}^{K}) \neq 0$, then there is a non-zero probability that $x_{\sample}$ is selected from $\mathbb{X}_{\text{initial}}^{T_{\max}} \setminus \mathbb{X}_{\tree}^{K}$. Let $T$ be the number of steps that is required to steer $x_{\sample}$ into $\mathbb{X}_{\text{goal}}$. Then by non-zero probability, $T \in \{1,\cdots,T_{\max} \}$ is chosen in Line 4:. Now we need to prove that the solution to Subproblem \ref{subproblem} with $\mathbb{X}_{\text{target}}=\mathbb{X}_{\text{goal}}$ returns a sequence of polytopes such that $\mathbb{X}_0$ has non-zero volume with non-zero probability. This fact follows from polyhedral partition of explicit hybrid MPC. We know that the explicit solution to \eqref{eq_MPC} with horizon $T$ produces a finite number of polyhedral partitions with affine feedback law in each one. Therefore, $x_{\sample}$ with probability $1$ belongs to the interior of one of the polyhedral partitions with non-zero volume. Let it be denoted by $\mathbb{P}_{\sample}$. Thus, the solution to Subproblem \ref{subproblem} is guaranteed to produce a non-zero volume $\mathbb{X}_0 \subseteq \mathbb{P}_{\sample}$ with  positive volume constraint- recall the upper/lower triangular restriction for $G_0$ mentioned in Sec. \ref{sec_volume}. 

For the case $\mathbb{X}_{\text{initial}} \setminus \mathbb{X}_{\tree}^{K}$ we replace $\mathbb{X}_{\text{goal}}$ by polytopes in $\mathbb{X}_{\text{initial}}^{T_{\max}}$ and we arrive a similar argument for $\mathbb{X}_{\text{initial}}^{\kappa T_{\max}}, \kappa=1,2,\cdots$, and the rest of the proof follows. 

Now we prove that we cover polyhedral partitions that have less than $n$ dimensions by relaxing the requirement that $\vol(\mathbb{X}_0)>0$. Let $\mathbb{P}_{<n}$ be such a polytope with dimension $q<n$. There is non-zero probability that $x_{\sample}$ is chosen from $\eta$-neighborhood of $\mathbb{P}_{<n}$ such that no other partition is intersecting with this neighborhood. Therefore, when solving Subproblem \ref{subproblem}, there is non-zero probability that we "land" on $\mathbb{P}_{<n}$, and obtain a $q$-dimensional local polytope (claim A) that its intersection with $\mathbb{P}_{<n}$ has non-zero measure in $q$-dimensional local coordinates, which completes the proof. 

(claim A): Our solver for Subproblem \ref{subproblem} is able to find polytopes with maximal dimension. In other words, we are able to maximize the rank of $G_0$. We already have shown that we can obtain full-rank $G_0$ with upper/lower triangular restriction. When $G_0$ can not be full rank, we can re-parametrize it by $G_0=G^f G^q$, where $G^f$ is a parameterized full rank square matrix and $G_q$ is a given matrix with rank $q<n$ such that $G^f$ exists for $q$ but not for $q+1$ - the value of $q$ can be obtained by line search. 
\end{proof}

\begin{remark}
In comparison to \cite{tedrake2010lqr}, we have dropped two key assumptions leading to probabilistic feedback coverage. First, it was assumed in \cite{tedrake2010lqr} that nonlinear trajectory optimization is always able to find trajectories connecting to the tree, if any exists, with non-zero probability. Since our trajectory optimization method is based on MICP and is complete, we do not require this assumption. Second, it was assumed in \cite{tedrake2010lqr} that it is always possible to obtain funnels with non-zero volume around any nominal trajectory. This assumption is not reasonable in our setting. But we know from explicit MPC that any "funnel" with maximum dimension, if exists, can already be obtained by local affine control laws, which our controller is already based on.   
\end{remark}

\begin{remark}
We have not discussed about the rate of convergence in Theorem 1. In practice, achieving reasonable feedback coverage can be very challenging. Moreover, there is no straightforward guidance to when to switch to polytopes with dimension less than $n$. We note that such polytopes occur in many interesting problems in contact-based robotics. For example, inelastic contacts lead to dimension reduction in the state of the system.
\end{remark}

%% file: root_pwa_icra.bbl
\begin{thebibliography}{10}
\providecommand{\url}[1]{#1}
\csname url@rmstyle\endcsname
\providecommand{\newblock}{\relax}
\providecommand{\bibinfo}[2]{#2}
\providecommand\BIBentrySTDinterwordspacing{\spaceskip=0pt\relax}
\providecommand\BIBentryALTinterwordstretchfactor{4}
\providecommand\BIBentryALTinterwordspacing{\spaceskip=\fontdimen2\font plus
\BIBentryALTinterwordstretchfactor\fontdimen3\font minus
  \fontdimen4\font\relax}
\providecommand\BIBforeignlanguage[2]{{%
\expandafter\ifx\csname l@#1\endcsname\relax
\typeout{** WARNING: IEEEtran.bst: No hyphenation pattern has been}%
\typeout{** loaded for the language `#1'. Using the pattern for}%
\typeout{** the default language instead.}%
\else
\language=\csname l@#1\endcsname
\fi
#2}}

\bibitem{tedrake2010lqr}
R.~Tedrake, I.~R. Manchester, M.~Tobenkin, and J.~W. Roberts, ``Lqr-trees:
  Feedback motion planning via sums-of-squares verification,'' \emph{The
  International Journal of Robotics Research}, vol.~29, no.~8, pp. 1038--1052,
  2010.

\bibitem{collins2005bipedal}
S.~H. Collins and A.~Ruina, ``A bipedal walking robot with efficient and
  human-like gait,'' in \emph{Robotics and Automation, 2005. ICRA 2005.
  Proceedings of the 2005 IEEE International Conference on}.\hskip 1em plus
  0.5em minus 0.4em\relax IEEE, 2005, pp. 1983--1988.

\bibitem{grizzle2014models}
J.~W. Grizzle, C.~Chevallereau, R.~W. Sinnet, and A.~D. Ames, ``Models,
  feedback control, and open problems of 3d bipedal robotic walking,''
  \emph{Automatica}, vol.~50, no.~8, pp. 1955--1988, 2014.

\bibitem{deits2014footstep}
R.~Deits and R.~Tedrake, ``Footstep planning on uneven terrain with
  mixed-integer convex optimization,'' in \emph{Humanoid Robots (Humanoids),
  2014 14th IEEE-RAS International Conference on}.\hskip 1em plus 0.5em minus
  0.4em\relax IEEE, 2014, pp. 279--286.

\bibitem{grizzle2009mabel}
J.~W. Grizzle, J.~Hurst, B.~Morris, H.-W. Park, and K.~Sreenath, ``Mabel, a new
  robotic bipedal walker and runner,'' in \emph{American Control Conference,
  2009. ACC'09.}\hskip 1em plus 0.5em minus 0.4em\relax IEEE, 2009, pp.
  2030--2036.

\bibitem{okamura2000overview}
A.~M. Okamura, N.~Smaby, and M.~R. Cutkosky, ``An overview of dexterous
  manipulation,'' in \emph{Robotics and Automation, 2000. Proceedings. ICRA'00.
  IEEE International Conference on}, vol.~1.\hskip 1em plus 0.5em minus
  0.4em\relax IEEE, 2000, pp. 255--262.

\bibitem{pratt2006capture}
J.~Pratt, J.~Carff, S.~Drakunov, and A.~Goswami, ``Capture point: A step toward
  humanoid push recovery,'' in \emph{Humanoid Robots, 2006 6th IEEE-RAS
  International Conference on}.\hskip 1em plus 0.5em minus 0.4em\relax IEEE,
  2006, pp. 200--207.

\bibitem{valenzuela2016mixed}
A.~K. Valenzuela, ``Mixed-integer convex optimization for planning aggressive
  motions of legged robots over rough terrain,'' Ph.D. dissertation,
  Massachusetts Institute of Technology, 2016.

\bibitem{marcucci2017approximate}
T.~Marcucci, R.~Deits, M.~Gabiccini, A.~Biechi, and R.~Tedrake, ``Approximate
  hybrid model predictive control for multi-contact push recovery in complex
  environments,'' in \emph{Humanoid Robotics (Humanoids), 2017 IEEE-RAS 17th
  International Conference on}.\hskip 1em plus 0.5em minus 0.4em\relax IEEE,
  2017, pp. 31--38.

\bibitem{mehr2017stochastic}
N.~Mehr, D.~Sadigh, R.~Horowitz, S.~S. Sastry, and S.~A. Seshia, ``Stochastic
  predictive freeway ramp metering from signal temporal logic specifications,''
  in \emph{American Control Conference (ACC), 2017}.\hskip 1em plus 0.5em minus
  0.4em\relax IEEE, 2017, pp. 4884--4889.

\bibitem{de2004qualitative}
H.~De~Jong, J.-L. Gouz{\'e}, C.~Hernandez, M.~Page, T.~Sari, and J.~Geiselmann,
  ``Qualitative simulation of genetic regulatory networks using
  piecewise-linear models,'' \emph{Bulletin of mathematical biology}, vol.~66,
  no.~2, pp. 301--340, 2004.

\bibitem{Yordanov2012}
\BIBentryALTinterwordspacing
B.~Yordanov, J.~Tumova, I.~Cerna, J.~Barnat, and C.~Belta, ``{Temporal Logic
  Control of Discrete-Time Piecewise Affine Systems},'' \emph{IEEE Transactions
  on Automatic Control}, vol.~57, no.~6, pp. 1491--1504, 2012. [Online].
  Available:
  \url{http://ieeexplore.ieee.org/lpdocs/epic03/wrapper.htm?arnumber=6104370}
\BIBentrySTDinterwordspacing

\bibitem{hassibi1998quadratic}
A.~Hassibi and S.~Boyd, ``Quadratic stabilization and control of
  piecewise-linear systems,'' in \emph{American Control Conference, 1998.
  Proceedings of the 1998}, vol.~6.\hskip 1em plus 0.5em minus 0.4em\relax
  IEEE, 1998, pp. 3659--3664.

\bibitem{han2017feedback}
W.~Han and R.~Tedrake, ``Feedback design for multi-contact push recovery via
  lmi approximation of the piecewise-affine quadratic regulator,'' in \emph{n
  Proceedings of the 2017 IEEE-RAS International Conference on Humanoid Robots,
  2017}.\hskip 1em plus 0.5em minus 0.4em\relax IEEE-RAS, 2017.

\bibitem{dua2000algorithm}
V.~Dua and E.~N. Pistikopoulos, ``An algorithm for the solution of
  multiparametric mixed integer linear programming problems,'' \emph{Annals of
  operations research}, vol.~99, no. 1-4, pp. 123--139, 2000.

\bibitem{bemporad2000piecewise}
A.~Bemporad, F.~Borrelli, and M.~Morari, ``Piecewise linear optimal controllers
  for hybrid systems,'' in \emph{American Control Conference, 2000. Proceedings
  of the 2000}, vol.~2.\hskip 1em plus 0.5em minus 0.4em\relax IEEE, 2000, pp.
  1190--1194.

\bibitem{hogan2016feedback}
F.~R. Hogan and A.~Rodriguez, ``Feedback control of the pusher-slider system: A
  story of hybrid and underactuated contact dynamics,'' \emph{arXiv preprint
  arXiv:1611.08268}, 2016.

\bibitem{hogan2017reactive}
F.~R. Hogan, E.~R. Grau, and A.~Rodriguez, ``Reactive planar manipulation with
  convex hybrid mpc,'' \emph{arXiv preprint arXiv:1710.05724}, 2017.

\bibitem{karaman2011sampling}
S.~Karaman and E.~Frazzoli, ``Sampling-based algorithms for optimal motion
  planning,'' \emph{The international journal of robotics research}, vol.~30,
  no.~7, pp. 846--894, 2011.

\bibitem{lavalle2006planning}
S.~M. LaValle, \emph{Planning algorithms}.\hskip 1em plus 0.5em minus
  0.4em\relax Cambridge university press, 2006.

\bibitem{branicky2006sampling}
M.~S. Branicky, M.~M. Curtiss, J.~Levine, and S.~Morgan, ``Sampling-based
  planning, control and verification of hybrid systems,'' \emph{IEE
  Proceedings-Control Theory and Applications}, vol. 153, no.~5, pp. 575--590,
  2006.

\bibitem{rajasekaran2017lqr}
S.~Rajasekaran, R.~Natarajan, and J.~D. Taylor, ``Towards planning and control
  of hybrid systems with limit cycle using lqr trees,'' in \emph{2017 IEEE/RSJ
  International Conference on Intelligent Robots and Systems (IROS)}, Sept
  2017, pp. 5196--5203.

\bibitem{manchester2011regions}
I.~R. Manchester, M.~M. Tobenkin, M.~Levashov, and R.~Tedrake, ``Regions of
  attraction for hybrid limit cycles of walking robots,'' \emph{IFAC
  Proceedings Volumes}, vol.~44, no.~1, pp. 5801--5806, 2011.

\bibitem{Bemporad1999}
A.~Bemporad and M.~Morari, ``{Control of systems integrating logic, dynamics,
  and constraints},'' \emph{Automatica}, vol.~35, no.~3, pp. 407--427, 1999.

\bibitem{link_github}
S.~Sadraddini and R.~Tedrake. https://github.com/sadraddini/pwa-control.

\bibitem{rakovic2007optimized}
S.~Rakovi{\'c}, E.~C. Kerrigan, D.~Q. Mayne, and K.~I. Kouramas, ``Optimized
  robust control invariance for linear discrete-time systems: Theoretical
  foundations,'' \emph{Automatica}, vol.~43, no.~5, pp. 831--841, 2007.

\bibitem{vandenberghe1998determinant}
L.~Vandenberghe, S.~Boyd, and S.-P. Wu, ``Determinant maximization with linear
  matrix inequality constraints,'' \emph{SIAM journal on matrix analysis and
  applications}, vol.~19, no.~2, pp. 499--533, 1998.

\bibitem{mete2012pattern}
H.~O. Mete and Z.~B. Zabinsky, ``Pattern hit-and-run for sampling efficiently
  on polytopes,'' \emph{Operations Research Letters}, vol.~40, no.~1, pp.
  6--11, 2012.

\bibitem{tondel2003evaluation}
P.~T{\o}ndel, T.~A. Johansen, and A.~Bemporad, ``Evaluation of piecewise affine
  control via binary search tree,'' \emph{Automatica}, vol.~39, no.~5, pp.
  945--950, 2003.

\bibitem{ansari2016sequential}
A.~R. Ansari and T.~D. Murphey, ``Sequential action control: Closed-form
  optimal control for nonlinear and nonsmooth systems.'' \emph{IEEE Trans.
  Robotics}, vol.~32, no.~5, pp. 1196--1214, 2016.

\bibitem{elbanhawi2014sampling}
M.~Elbanhawi and M.~Simic, ``Sampling-based robot motion planning: A review,''
  \emph{Ieee access}, vol.~2, pp. 56--77, 2014.

\bibitem{janson2015fast}
L.~Janson, E.~Schmerling, A.~Clark, and M.~Pavone, ``Fast marching tree: A fast
  marching sampling-based method for optimal motion planning in many
  dimensions,'' \emph{The International journal of robotics research}, vol.~34,
  no.~7, pp. 883--921, 2015.

\end{thebibliography}
